\documentclass[11pt]{article}

\usepackage{graphicx, mathtools, amssymb, latexsym, amsmath, amsfonts, amsthm, bbm, tikz}
\mathtoolsset{showonlyrefs}

\usepackage{hyperref}
\hypersetup{
    colorlinks=true, 
    linkcolor=blue, 
    urlcolor=red, 
    linktoc=all 
}

\usepackage{xcolor}
\definecolor{DarkGreen}{rgb}{0.1,0.5,0.1}
\definecolor{DarkRed}{rgb}{0.5,0.1,0.1}
\definecolor{DarkBlue}{rgb}{0.1,0.1,0.5}

\usepackage{array, fancyhdr, bm, xcolor}

\theoremstyle{plain}
\newtheorem{theorem}{Theorem}[section]
\newtheorem{lemma}[theorem]{Lemma}
\newtheorem{claim}[theorem]{Claim}

\theoremstyle{definition}
\newtheorem{definition}[theorem]{Definition}
\newtheorem{remark}[theorem]{Remark}

\DeclareMathOperator*\supp{supp}
\newcommand\defeq{\ensuremath{\stackrel{\rm def}{=}}} 
\newcommand{\ind}[1]{^{(#1)}}

\begin{document}
\title{Improved batch code lower bounds}
\author{Ray Li\thanks{Department of Computer Science, Stanford University. rayyli@cs.stanford.edu. Research supported by NSF Grants DGE-1656518, CCF-1814629, and by Jacob Fox's Packard Fellowship.} \and Mary Wootters\thanks{Departments of Computer Science and Electrical Engineering, Stanford University.  marykw@stanford.edu.  Research partially supported by NSF Grant CCF-1844628 and by a Sloan Research Fellowship.}}
\date{\today}
\maketitle

\begin{abstract}
  Batch codes are a useful notion of locality for error correcting codes, originally introduced in the context of distributed storage and cryptography.  Many constructions of batch codes have been given, but few lower bound (limitation) results are known, leaving gaps between the best known constructions and best known lower bounds. Towards determining the optimal redundancy of batch codes, we prove a new lower bound on the redundancy of batch codes. Specifically, we study (primitive, multiset) \em linear \em batch codes that systematically encode $n$ information symbols into $N$ codeword symbols, with the requirement that any multiset of $k$ symbol requests can be obtained in disjoint ways. We show that such batch codes need $\Omega(\sqrt{Nk})$ symbols of redundancy, improving on the previous best lower bounds of $\Omega(\sqrt{N}+k)$ at all $k=n^\varepsilon$ with $\varepsilon\in(0,1)$. Our proof follows from analyzing the dimension of the order-$O(k)$ tensor of the batch code's dual code.
\end{abstract}

\section{Introduction}

In this work, we study \emph{batch codes}, a notion of locality for error correcting codes, and show stronger limitations on batch codes for almost all parameter regimes. 

Batch codes were introduced in the context of load-balancing in distributed storage and private information retrieval in cryptography \cite{IKOS04}.
Informally, a (primitive, multiset) $k$-batch code is a error correcting code $C:\Sigma^n\to \Sigma^N$,  mapping $n$ information symbols to $N$ codeword symbols, such that every multiset of $k$ information symbols can be recovered from $k$ pairwise disjoint recovering sets.
In constructing batch codes, we would like the \emph{locality} parameter $k$ to be as large as possible.
On the other hand, we would like to minimize the \emph{redundancy} $N-n$ of our code, representing the number of redundant bits in our encoding.
In this work, we prove new limitations on the quantitative tradeoff between the locality and the redundancy of batch codes.
Formally, a batch code is defined as follows.
\begin{definition}\label{def:batch}
  Let $C:\Sigma^n\to\Sigma^N$ be a code that maps $x_1,\dots,x_n$ to $c_1,\dots,c_N$.
  The code $C$ is a \emph{$k$-batch code} if, for every multiset of indices $\{i_1,\dots,i_k\}\subset[n]$, there exist $k$ mutually disjoint sets $R_1,\dots,R_k\subset[N]$ and functions $g_1,\dots,g_k$ such that for all information symbols $x_1,\dots,x_n\in\Sigma$ and codewords $(c_1,\dots,c_N)=C(x_1,\dots,x_n)\in\Sigma^N$ and for all $j\in[k]$, we have $g_j(c|_{R_j})=x_{i_j}$.
  We say $C$ is a \emph{linear} batch code if $\Sigma$ is a finite field, the functions $g_i$ are all linear, and $C$ is a linear map.
  We say $C$ has a \emph{systematic encoding} if $c_i=x_i$ for $i=1,\dots,n$. 
\end{definition}
\begin{remark}[Primitive Multiset Batch Codes]\label{rem:primitive}
We note that our definition of a batch code here is a refinement of the standard, more general notion of batch codes, introduced by \cite{IKOS04}.
In that definition of a batch code, $n$ symbols are encoded into ``buckets'' of symbols such that the total size of all buckets is $N$, and each batch of $k$ symbols can be decoded by reading at most one symbol from each bucket. A batch code is a \emph{multiset} batch code if (a) the $k$ symbols can form a multiset and (b) the $k$ symbols can be decoded by querying $k$ pairwise disjoint sets of buckets.
When each bucket can store a single symbol, the multiset batch code is said to be \emph{primitive.}  Because we only focus on primitive multiset batch codes in this work, we drop the adjectives ``primitive'' and ``multiset'' throughout and simply use ``batch code,'' as per Definition~\ref{def:batch}.
\end{remark}

The main goal in constructing batch codes is to determine the minimum redundancy $r(n,k)$ of a $k$-batch code encoding $n$ symbols. 
Many works \cite{IKOS04, DGRS14, HS15, VY16, AY17, PV19a, PPV20, HPPV20, HPPVY20}, have constructed batch codes, giving good upper bounds on $r(n,k)$. 
Figure~\ref{fig:main} and Section~\ref{ssec:related} give a summary of the known constructions.

On the other hand, few limitations are known on the optimal locality versus redundancy tradeoff $r(n,k)$.
An easy lower bound shows that a $k$-batch code has minimum Hamming distance $k$, and thus must have redundancy at least $k$.
The only nontrivial lower bound on the redundancy of (linear) batch codes is given by \cite{RV16,Woo16}, who showed that a $k$-batch code has redundancy at least $\Omega(\sqrt{n})$ when $k\ge 3$, and this is tight up to a logarithmic factor for constant $k$ \cite{VY16}.
A priori, given these two lower bounds, and the fact that there exist $n^\varepsilon$ batch codes with redundancy $n^{\delta}$ for every $\varepsilon\in(0,1)$ and for $\delta=\delta(\varepsilon)<1$ \cite{HPPV20}, it seemed possible that the optimal $\delta$ could match the best known lower bounds at $\max(1/2,\varepsilon)$. 
In our work, we refute this possibility for all $\varepsilon\in(0,1)$ (under two reasonable assumptions that also appear in prior lower bounds).

\begin{theorem}
  \label{thm:main}
  A linear $k$-batch code of length $N$ with a systematic encoding must have $\Omega(\sqrt{Nk})$ redundancy.
\end{theorem}
Though our lower bound assumes linearity and a systematic encoding, we point out that the only other nontrivial lower bounds for batch codes \cite{RV16,Woo16} also assume linearity, with \cite{Woo16} additionally assuming systematic encoding, and with \cite{RV16} handling non-systematic encoding but additionally assuming that the code is binary (our lower bound works for linear batch codes over any field).
We also point out that many existing constructions of batch codes are linear and have systematic encoding.

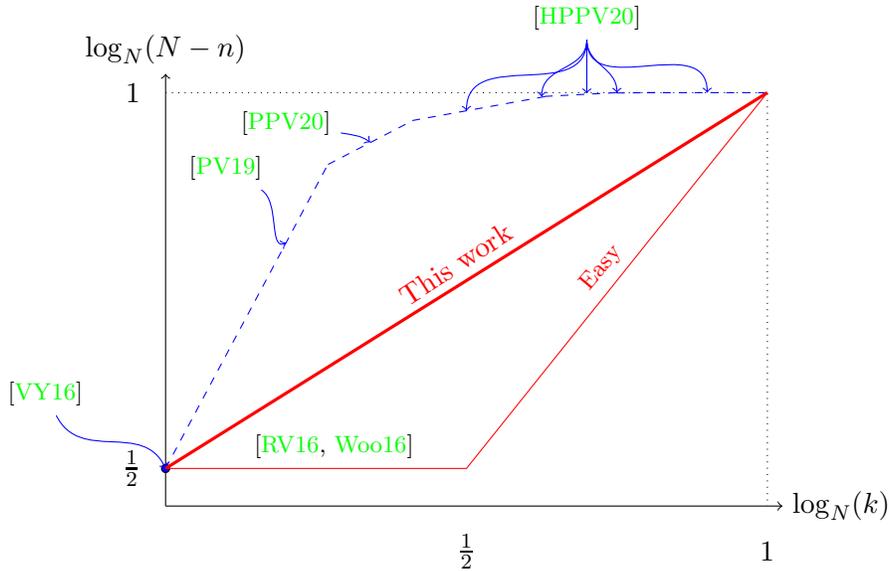
\begin{figure}
\label{fig:main}
\begin{center}
\begin{tikzpicture}[xscale=8, yscale=10]
\draw[->] (0,.45) to (1.025,.45);
\draw[->] (0,.45) to (0,1.025);
\node[anchor=west] at (1.025,.45) {$\log_N(k)$};
\node[anchor=south] at (0,1.025) {$\log_N(N - n)$};
\draw[dotted] (0,1) -- (1,1) -- (1,0.45);
\node[anchor=east] at (-.025, .5) {$\frac{1}{2}$};
\node[anchor=east] at (-.025, 1) {$1$};
\node at (1,.39) {1};
\node at (.5, .39) {$\frac{1}{2}$};
\node[draw, fill=blue, circle,scale=0.3](a) at (0, 0.5) {};
\node (vy16) at (-0.2,0.6) {\footnotesize \cite{VY16}};
\draw[blue,->] (vy16) to [out=-60,in=120] (0,0.5);
\draw[domain=0:0.27, smooth, variable=\x, blue,dashed] plot ( {\x} , {.5 + \x * 1.5});
\node (pv19) at (0.1,0.9) {\footnotesize \cite{PV19a}};
\draw[blue,->] (pv19) to [out=-30,in=200] (0.2,0.8);
\draw[domain=0.27:0.41, smooth, variable=\x, blue,dashed] plot ( {\x} , {log10(3)/log10(4) + \x * (2-log10(3)/log10(2))});
\node (ppv20) at (0.2,0.96) {\footnotesize \cite{PPV20}};
\draw[blue,->] (ppv20) to [out=350,in=150] (0.34,0.934);
\node (hppv20) at (0.7,1.1) {\footnotesize \cite{HPPV20}};
\foreach \m / \lam /\start / \end in {3/7.2361/0.41/0.64,4/15.5436/0.63/0.75,5/31.7877/0.75/0.8,6/63.9217/0.8/0.83}{
  \draw[domain=\start:\end, smooth, variable=\x, blue,dashed] plot ( {\x} , {(\m-(log10(\lam)/log10(2)))*\x + ((\m-1) * (log10(\lam)/log10(2)) / \m) - \m + 2});
  \foreach \x in {(2*\m-3)/(2*\m)} {
    \coordinate (a\m) at ( {\x} , {(\m-(log10(\lam)/log10(2)))*\x + ((\m-1) * (log10(\lam)/log10(2)) / \m) - \m + 2});
    \draw[blue,->] (hppv20) to [out=270,in=90] (a\m);
  }
}
\draw[blue,dashed] (0.8,1)--(1,1);
\draw[blue,->] (hppv20) to [out=270,in=90] (0.9,1);
\draw[red] (0,.5) -- (0.5, .5) -- (1,1);
\draw[very thick, red] (0,.5)  -- (1,1);
\node[anchor=south] at (0.28,0.5) {\footnotesize\cite{RV16,Woo16}};
\node[anchor=south,rotate=50,red] at (0.75,0.75) {\footnotesize Easy};
\node[anchor=south,red,rotate=33](us) at (.5,0.75) {This work};
\end{tikzpicture}
\end{center}
\caption{$r(n,k)$, the minimum possible redundancy of a $k$-batch code encoding $n$ information symbols. Blue dashed segments indicate code constructions. Red solid segments indicate lower bounds for linear batch codes. Our lower bound additionally assumes systematic encoding. 
}
\end{figure}

\subsection{Related work}
\label{ssec:related}

\paragraph{Prior constructions of batch codes.}

We would like to understand $r(n,k)$, the minimum possible redundancy of a $k$-batch code encoding $n$ information symbols.
The following upper bounds on $r(n,k)$ are known, with the best known ones illustrated in Figure~\ref{fig:main}, along with the known lower bounds.
\begin{itemize}
\item $r(n,k)\le O(k^{4})$, for $k=n^\varepsilon$ with $1/5<\varepsilon\le 7/32$ \cite{DGRS14}
\item $r(n,n^{1/4})\le n^{7/8}$ \cite{DGRS14}
\item $r(n,k)\le \tilde O(\sqrt{n})$ for any fixed $k$ \cite{VY16}.
\item $r(n,n^\varepsilon)\le O(n^{2/3 + 5\varepsilon/3})$ for $\varepsilon<1/2$ \cite{AY17}.
\item $r(n,n^\varepsilon)\le O(n^{5/6 + \varepsilon/3})$ for $\varepsilon<1/2$ \cite{AY17}.
\item $r(n,n^\varepsilon)\le \tilde O(n^{(3\varepsilon+1)/2})$ for $0<\varepsilon<1/3$ \cite{PV19a}.
\item $r(n,n^\varepsilon)\le \tilde O(n^{\log_4(3) + (2-\log_2(3))\varepsilon})$ for $0<\varepsilon<1/2$ \cite{PPV20}.
\item $r(n,n^{\varepsilon})\le O(n^{\delta})$ for $0\le \varepsilon< 1$, where $\delta=\delta(\varepsilon)<1$ \cite{HPPV20}
\end{itemize}

\paragraph{Other notions of locality.} Batch codes are related to several other notions of locality.
Two closely related notions are private information retrieval (PIR) codes \cite{FVY15, AY17} and codes with the disjoint repair group property (DRGP) \cite{GW13, FischerGW17, LW19,HKLW21}, which are relaxations of batch codes that only require $k$ disjoint repair groups for a single information symbol (codeword symbol in the case of DRGP).
In the case that the locality parameter $k$ is linear in the block length, PIR codes and DRGP codes are in fact equivalent to constant query locally decodable codes (LDCs) and locally correctable codes (LCCs), respectively.
Other related notions include Locally Repairable Codes \cite{GHSY12}, 
LRCs with availability \cite{RPDV14}, batch codes with availability \cite{ZS16} where the repair group sizes are also bounded, switch codes \cite{WSCB13, WKC15, CGTZ15}, which are a special case of batch codes, and combinatorial batch codes \cite{PSW09}, which are a special case of (non-primitive) multiset batch codes.
In these settings, determining the optimal locality versus redundancy tradeoff is an interesting question.
We hope our techniques could be useful for proving lower bounds for some of these other notions of locality.
For more details about some of these other notions, we refer the reader to the survey \cite{Ska16}.

\subsection{Preliminaries}
\paragraph{Basic notation.}
We use $[N]$ to denote the set $\{1, \ldots, N\}$.  
For a vector $c\in\mathbb{F}^N$, let $\supp(c)\defeq \{i\in[N]:c_i\neq 0\}$. 
For two vector spaces $V,W$ over $\mathbb{F}$, we use $V\le W$ to denote $V$ is a subspace of $W$.

By abuse of notation, we identify a code $C:\mathbb{F}^n\to\mathbb{F}^N$ with a systematic encoding by its image $C\le\mathbb{F}^N$, which is a subspace of $\mathbb{F}^N$.
In the rest of this paper, all codes have systematic encoding and are represented this way.

\paragraph{Dual codes.}
The \emph{dual code} $C^\perp\le \mathbb{F}^N$ of a linear code $C\le\mathbb{F}^N$ is the subspace of all codewords orthogonal to every codeword of $C$.
One can easily check that $\dim C^\perp = N - \dim C$, so in particular the dimension of $C^\perp$ is equal to the redundancy of $C$.
We call elements of $C^\perp$ \emph{dual codewords (of $C$)}.
We use the following standard fact of dual codes, and include a proof for completeness.
\begin{lemma}
Suppose $C\le\mathbb{F}^N$ is a linear code with a systematic encoding, and there exists an index $i\in[N]$, a set $R\subset[N]\setminus\{i\}$, and a linear function $g$, such that $g(c|_R) = c_i$ for all codewords $c\in C$.
Then there exists a nonzero dual codeword $c^\perp$ such that $\supp(c^\perp) \subset R\cup \{i\}$ and $i\in \supp(c^\perp)$.
\label{lem:dual}
\end{lemma}
\begin{proof}
  Suppose $g$ is the function $g(c|_R) = \sum_{j\in R}^{} \alpha_j c_j$.
  Then $-c_i + \sum_{j\in R}^{} \alpha_j c_j = 0$ for all codewords $c$.
  Thus, the vector $c^\perp$ with $c^\perp_j=-1$ if $j=i$, $c^\perp_j=\alpha_j$ if $j\in R$, and $c^\perp_j=0$ otherwise, is a dual codeword. Vector $c^\perp$ also satisfies $\supp(c^\perp) \subset R\cup \{i\}$ and $i\in \supp(c^\perp)$ by construction, as desired.
\end{proof}

\paragraph{Tensor products.}
We let $e_i$ denote the standard basis vector in $\mathbb{F}^N$, so that $(e_i)_i = 1$ and $(e_i)_j = 0$ for all $j \neq i$.  
For $v^{(1)}, \ldots, v^{(s)} \in \mathbb{F}^N$, we define the tensor product $v^{(1)} \otimes \cdots \otimes v^{(s)} \in \mathbb{F}^{N^s}$ to be the vector indexed by tuples $(i_1, \ldots, i_s) \in [N]^s$ with $$(v^{(1)} \otimes \cdots \otimes v^{(s)})_{(i_1, \ldots, i_s)} = \prod_{j=1}^s v^{(j)}_{i_j}.$$
A tensor of this form is called a \emph{simple tensor}; more generally a \emph{tensor} is any linear combination of simple tensors. We note that the set of simple tensors
\[ \left\{\bigotimes_{j=1}^s e_{i_j} \,:\, (i_1, \ldots, i_s) \in [N]^s \right\} \]
forms a basis for $\mathbb{F}^{N^s}$.  We refer to this as the standard tensor basis for $\mathbb{F}^{N^s}$.  
Accordingly, every vector in $\mathbb{F}^{N^s}$ can be written as a linear combination of the standard tensor basis, which we call the \emph{standard basis representation}.
Because these tensors form a basis, we have the following useful fact.

\begin{lemma}
  Let $(e\ind{1},w\ind{1}),\dots,(e\ind{D},w\ind{D})$ be pairs of tensors in $\mathbb{F}^{N^s}$ such that, for all $i=1,\dots,D$, tensor $e\ind{i}$ is in the standard tensor basis and the standard basis representation of $w\ind{i}$ contains $e\ind{i}$ and none of $e\ind{i+1},\dots,e\ind{D}$.
  Then $w\ind{1},\dots,w\ind{D}$ are linearly independent.
\label{lem:tensor}
\end{lemma}
\begin{proof}
  Suppose for contradiction $\sum_{i=1}^{D} \alpha_i w\ind{i} = 0$ is a nonzero linear combination of $w\ind{1},\dots,w\ind{D}$.
  Let $j$ be the largest index such that $\alpha_{j}\neq 0$.
  By definition of $w\ind{j}$, we may write $w\ind{j} = \beta e\ind{j} + w'$ for some $\beta\neq 0$ and some $w'\in \mathbb{F}^{N^s}$ not containing $e\ind{j}$ in its standard basis representation.
  Then the coefficient of $e\ind{j}$ in the standard basis representation of $\sum_{i=1}^{D} \alpha_i w\ind{i} = (\sum_{i < j}^{} \alpha_i w\ind{i} + \alpha_j w') + \alpha_j\beta e\ind{j}$ is exactly $\alpha_j\beta$, since $w\ind{1},\dots,w\ind{j-1}$ and $w'$ do not contain $e\ind{j}$ in their standard basis representation, so $\alpha_j\beta=0$.
  This contradicts that $\alpha_j,\beta\neq 0$.
\end{proof}

For a subspace $V\le \mathbb{F}^N$ and integer $s$, we let $V^{\otimes s}$ denote the subspace of $\mathbb{F}^{N^s}$ spanned by simple tensors $v\ind{1}\otimes \cdots\otimes v\ind{s}$ for $v\ind{1},\dots,v\ind{s}\in V$.
A standard fact says that $\dim V^{\otimes s} = (\dim V)^s$, for all subspaces $V$ and positive integers $s$.

\section{Proof of the main theorem}
\subsection{Sketch of the proof}

Though our proof is short, we provide a brief summary to highlight the main ideas.
For intuition in this sketch, we strengthen the definition of batch codes to require that the repair groups $R_1,\dots,R_k$ are not only pairwise disjoint but also that they are each disjoint from $\{i_1,\dots,i_k\}$. 
Removing this assumption is not difficult.

Suppose $k=2t$ for some integer $t$.
Let $V = C^\perp$.
Call a simple tensor $e_{i'_1}\otimes \cdots\otimes  e_{i'_{2t}}\in\mathbb{F}^{N^{2t}}$ \emph{good} if the multiset $\{i'_1,\dots,i'_{2t}\}\subset[n]$ contains exactly $t$ distinct elements, each appearing exactly twice. 
Given such a multiset $\{i'_1, i'_2, \ldots, i'_{2t}\}$ for a good tensor, the definition of a $2t$-batch code together with Lemma~\ref{lem:dual} guarantee the existence of $2t$ dual codewords, one for each $i'_j$, so that the support of the corresponding codeword is $R_j \cup \{i'_j\}$, and where the $R_j$ are all disjoint from each other and $\{i_1',\dots,i_{2t}'\}$.  
By tensoring these dual codewords, we obtain a tensor $w_{i'_1, \ldots, i'_{2t}} \in V^{\otimes 2t}$ whose standard basis representation contains the good simple tensor $\bigotimes_{j=1}^{2t} e_{i'_j}$ and, crucially, no other good simple tensors.

Because these tensors $w_{i'_1, \ldots, i'_{2t}}$ have only one good simple tensor each, unique to them, they must be linearly independent.  Thus, the dimension of $V^{\otimes 2t}$ is lower bounded by the number of tensors $w_{i'_1,\dots,i'_{2t}}$, which is at least the number of good simple tensors, which we can count to be $\binom{n}{t}\binom{2t}{2,2,\dots,2} \ge \Omega(nt)^t$.
As $\dim(V^{\otimes 2t}) = (\dim V)^{2t}$, we conclude that $\dim V \ge (\Omega(nt)^t)^{1/2t} = \Omega(\sqrt{nt}) = \Omega(\sqrt{Nk})$, as desired.

\subsection{Full proof}
\begin{proof}[Proof of Theorem~\ref{thm:main}]
  Let $t=\lfloor k/3\rfloor$.
  Define vector spaces $V\defeq C^\perp$ and $W \defeq V^{\otimes 2t}$. 
  In the standard tensor basis of $\mathbb{F}^{N^{2t}}$, call a simple tensor $e_{i'_1}\otimes \cdots\otimes  e_{i'_{2t}}$ \emph{good} if the multiset $\{i'_1,\dots,i'_{2t}\}$ contains exactly $t$ distinct elements, each appearing exactly twice.

  For every $t$-tuple $(i_1,\dots,i_{t})\in[n]^t$ with $i_1>\cdots>i_t$, consider the $t$-multiset of symbols $\cup_{j=1}^{t} \{x_{i_j},x_{i_j},x_{i_j}\}$.
  Since $C$ is a $3t$-batch code, these have recovery sets $R_{j,1},R_{j,2},R_{j,3}$ for $j=1,\dots,t$ that are pairwise disjoint.
  Furthermore, for any $j$, at least two of the $R_{j,1},R_{j,2},R_{j,3}$ do not contain $\{i_j\}$. 
  Thus, for each $j$, by Lemma~\ref{lem:dual}, there exist two dual codewords $c\ind{j,1},c\ind{j,2}$ such that $\supp(c\ind{j,1})\cap \supp(c\ind{j,2}) = \{i_j\}$.
  Define $\mathcal{D}_{i_1,\dots,i_t}\defeq (c\ind{j,\ell})_{1\le j\le t, \ell=1,2}$.

  Call a bijection $\pi:[2t]\to [t]\times [2]$ a \emph{good map}.
  For a good map $\pi$, let $\pi_1:[2t]\to [t]$ denote the first coordinate of $\pi$ and let $\pi_2:[2t]\to [2]$ denote the second coordinate of $\pi$.
  For each good map $\pi$ and each $i_1,\dots,i_t$, let $e_{i_1,\dots,i_t,\pi}$ denote the simple tensor
  \begin{align}
    e_{i_1,\dots,i_t,\pi}\defeq \bigotimes_{j=1}^{2t}e_{i_{\pi_1(j)}}.
  \end{align}
  Clearly $e_{i_1,\dots.i_t,\pi}$ is a good simple tensor, since $\pi$ is a bijection and thus each $\pi_1(j)$ appears twice.

  For each good map $\pi$ and each $i_1,\dots,i_t$, if $\mathcal{D}_{i_1,\dots,i_t}=(c\ind{j,\ell})_{1\le j\le t, \ell=1,2}$, then let $w_{i_1,\dots,i_t,\pi}$ denote the tensor
  \begin{align}
    \label{eq:w}
    w_{i_1,\dots,i_t,\pi}\defeq \bigotimes_{j=1}^{2t} c\ind{\pi_1(j),\pi_2(j)}.
  \end{align}
  As $c\ind{\pi_1(j),\pi_2(j)}$ are all in $V$, each $w_{i_1,\dots,i_t,\pi}$ is in $W=V^{\otimes 2t}$.
  \begin{claim}
    \label{cl:2}
    Each $w_{i_1,\dots,i_t,\pi}$, written in the standard basis of $\mathbb{F}^{N^{2t}}$ has at most $3^t$ good simple tensors. 
  \end{claim}
  \begin{proof}
    Note that $\supp(c\ind{j,\ell}) \subseteq \{i_j\}\cup R_{j,\ell}$, and furthermore that $R_{j,\ell}$ are pairwise disjoint. Thus, each of $i_1,\dots,i_t$ appears in at most three supports $\supp(c\ind{j,\ell})$; that is, $i_j$ can appear in $\supp(c\ind{j,1})$, $\supp(c\ind{j,2})$, and then in $\supp(c\ind{j',\ell})$ for at most one other $(j', \ell)$, so that $i_j \in R_{j', \ell}$.  Further, any element $i \in [N] \setminus \{i_1, \ldots, i_t\}$ appears in at most one set $R_{j, \ell}$ and hence in at most one support $\supp(c\ind{j,\ell})$. 
    Thus, any simple tensor $e_{i'_1}\otimes\cdots\otimes e_{i'_{2t}}$ appearing in the standard basis representation of $w_{i_1,\dots,i_t,\pi}$ has at most three of any of $e_{i_1},\dots,e_{i_t}$ and at most one any other $e_i$ in the product.
    Thus, in any \emph{good} simple tensor $e_{i'_1}\otimes\cdots\otimes e_{i'_{2t}}$ in the standard basis representation of $w_{i_1,\dots,i_t,\pi}$, there must be two of each of $e_{i_1},\dots,e_{i_t}$ in the product.
    For each $e_{i_j}$, there are at most three of the $2t$ positions where it can appear in the standard basis representation: in the (two) positions $j'$ such that $\pi(j')_1=j$, or the (at most one) position $j'$ such that $R_{\pi(j')_1,\pi(j')_2}$ contains $\{i_j\}$.
    Thus, there are at most $\binom{3}{2}=3$ choices for the positions of $e_{i_j}$ for each $j$, so there are at most $3^t$ good tensors in $w_{i_1,\dots,i_t}$.
  \end{proof}
  \begin{claim}
    \label{cl:3}
    Each $w_{i_1,\dots,i_t,\pi}$, written in the standard basis of $\mathbb{F}^{N^{2t}}$ contains the simple tensor $e_{i_1,\dots,i_t,\pi}$.
  \end{claim}
  \begin{proof}
    For each $j=1,\dots,2t$, the dual codeword $c\ind{\pi_1(j),\pi_2(j)}$ has a nonzero coefficient in coordinate $i_{\pi_1(j)}$, so $w_{i_1,\dots,i_t,\pi}=\bigotimes_{j=1}^{2t}c^{(\pi_1(j),\pi_2(j))}$ has a nonzero coefficient in coordinate $(i_{\pi_1(1)},i_{\pi_1(2)},\dots,i_{\pi_1(2t)})$. 
    Hence, the  standard basis representation of $w_{i_1,\dots,i_t,\pi}$ contains the simple tensor $\bigotimes_{j=1}^{2t} e_{i_{\pi_1(j)}} = e_{i_1,\dots,i_t,\pi}$.
  \end{proof}

  Now consider the set $E$ of all good simple tensors $e_{i_1,\dots,i_t,\pi}$ with $i_1>\cdots>i_t$ and $\pi$ a good map.
  Note that there are exactly $\binom{n}{t}\binom{2t}{2,2,\dots,2}$ elements of $E$, as we may first choose $i_1>\cdots>i_t$ and then choose $\pi_1$.  (Notice that the definition of $e_{i_1, \ldots, i_t, \pi}$ does not depend on $\pi_2$).
  Choose a sequence $(e\ind{1},w\ind{1}),(e\ind{2},w\ind{2}),\dots,\in E\times W$ as follows: given $(e\ind{1},w\ind{1}),\dots,(e\ind{r},w\ind{r})$, choose $e\ind{r+1}=e_{i_1,\dots,i_t,\pi}$ to be a good simple tensor not in the standard basis representation of any of $w\ind{1},\dots,w\ind{r}$, and let $w\ind{r+1}\defeq w_{i_1,\dots,i_t,\pi}$, which contains $e\ind{r+1}$ by Claim~\ref{cl:3}.
  By Claim~\ref{cl:2}, since each $w\ind{r}$ has at most $3^t$ good simple tensors in the standard basis representation, this process can be continued for at least $D\defeq |E|/3^t$ steps.
  Furthermore, we guarantee that $w\ind{r}$ contains $e\ind{r}$ and none of $e\ind{r+1},\dots,e\ind{D}$ for any $r=1,\dots,D$.
  Thus, by Lemma~\ref{lem:tensor}, the tensors $w\ind{1},\dots,w\ind{D}$ are linearly independent.
  Hence, we have
  \begin{align}
    (\dim V)^{2t} 
    =   \dim W
    \ge D
    =   |E|/3^t 
    =   \binom{n}{t}\cdot\frac{(2t)!}{2^t}\cdot\frac{1}{3^t} 
    \ge \left( \frac{n}{t} \right)^t\cdot\frac{(2t)^{2t}/e^{2t}}{2^t\cdot 3^t} 
    =  \frac{n^tt^t }{(3e^2/2)^t}.
  \label{}
  \end{align}
  Hence, the redundancy is at least
  \begin{align}
    \dim V \ge \Omega(\sqrt{nt}) \ge \Omega(\sqrt{Nk}).
  \end{align}
  Finally, we note that if $n$ is not $\Omega(N)$, then we are automatically done since the redundancy is at least $N/2$.  This completes the proof.
\end{proof}

\bibliographystyle{alpha}
\bibliography{bib}
\end{document}